\documentclass[a4paper,onecolumn,nonatbib,accepted=2026-05-29]{quantumarticle}
\pdfoutput=1
\usepackage[pdftex]{graphicx}
\graphicspath{{./figures/}} 
\usepackage{amsmath, amssymb,amsthm}
\usepackage{mathtools}
\usepackage[shortlabels]{enumitem}
\usepackage[dvipsnames]{xcolor}
\usepackage{cite}
\usepackage{upref} 
\usepackage{hyperref}
\hypersetup{colorlinks,citecolor=blue,filecolor=blue,linkcolor=blue,urlcolor=navyblue}
\definecolor{navyblue}{rgb}{0.0, 0.0, 0.5}

\newtheorem{thm}{Theorem}[section]
\newtheorem{prop}[thm]{Proposition}
\newtheorem{lem}[thm]{Lemma}
\newtheorem{cor}{Corollary}[section]
\newtheorem{conj}{Conjecture}
\theoremstyle{definition}

\newtheorem{rmk}{Remark}[section]

\newtheorem*{claim*}{Claim}

\newtheorem*{conj*}{Conjecture}

\newcommand{\C}{\mathbb{C}}

\newcommand{\E}{\mathbf{E}}
\newcommand{\oneb}{\mathbf{1}}
\renewcommand{\P}{\mathbf{P}}

\newcommand{\CN}{\mathcal{CN}}
\newcommand{\TV}{\mathrm{TV}}
\newcommand{\G}{\mathcal{G}}

\newcommand{\ggtnk}{\G\G^T_{NK}}
\newcommand{\gsym}{\G_N^\mathrm{sym}}
\newcommand{\dsim}{\sim}
\newcommand{\deq}{\overset{d}{=}}
\newcommand{\q}{N}
\newcommand{\Uqq}{\mathbf{A}}
\newcommand{\Gqq}{\mathbf{G}}
\newcommand{\bpq}{U_{pq}} 

\newcommand{\dkl}{D_{\mathrm{KL}}}

\newcommand{\Tr}{\operatorname{Tr}}

\newcommand{\Wg}{\operatorname{Wg}}
\newcommand{\iz}{\mathbf{I}_0}
\newcommand{\Haf}{\operatorname{Haf}}

\newcommand{\sharpp}{{\normalfont\fontfamily{lmss}\selectfont \#P}}

\newcommand{\poly}{\operatorname{poly}}

\newcommand\numberthis{\stepcounter{equation}\tag{\theequation}}
\numberwithin{equation}{section}

\begin{document}

\title{Proof of hiding conjecture in Gaussian boson sampling}

\author{Laura Shou}
\affiliation{Joint Quantum Institute, Department of Physics, University of Maryland, College Park, Maryland 20742, USA}  

\author{Sarah H. Miller}
\affiliation{Applied Research Laboratory for Intelligence and Security, University of Maryland, College Park, Maryland 20742, USA}

\author{Victor Galitski}
\affiliation{Joint Quantum Institute, Department of Physics, University of Maryland, College Park, Maryland 20742, USA}  

\begin{abstract}
Gaussian boson sampling (GBS) is a promising protocol for demonstrating quantum computational advantage. One of the key steps in the argument for classical hardness of GBS is the so-called ``hiding conjecture'', which asserts that one can ``hide'' a complex Gaussian matrix as a submatrix of the symmetric product of Haar unitary submatrices in total variation distance. In this paper, we  prove the hiding conjecture for input states with all input modes squeezed, which is a setup that has recently been realized experimentally [Madsen et al., Nature {\bf 606}, 75 (2022)]. In this setting, the hiding conjecture states that a $o(\sqrt{M})\times o(\sqrt{M})$ submatrix of an $M\times M$ circular orthogonal ensemble (COE) random matrix can be well-approximated by a complex symmetric Gaussian matrix in total variation distance as $M\to\infty$. This is the first rigorous proof of the hiding property for GBS in the experimentally relevant regime, and puts the argument for hardness of classically simulating GBS with the maximum number of squeezed input modes on a comparable level to that of the conventional boson sampling of [Aaronson and Arkhipov, Theory Comput. {\bf 9}, 143 (2013)].
\end{abstract}
\maketitle

\section{Introduction and main results}

The promise of quantum advantage is a key motivation behind quantum computing. However, practically demonstrating quantum advantage with physical quantum computers remains a technological challenge.
In this regard, Gaussian boson sampling (GBS), based on the boson sampling protocol of \cite{aa}, has emerged as a promising method for experimentally realizing this advantage
\cite{LundGBS2014, RahimiKeshariGBS2015, HamiltonGBS2017, KruseGBS2019, deshpande2022quantum,hangleiter2023computational, zhong2020quantum, Zhong2021, madsen2022, deng2023gaussian, liu2026gaussian}.
Both boson sampling and Gaussian boson sampling involve preparing photons in an initial state---a Fock state of $N$ photons for conventional, or Fock, boson sampling, and a Gaussian state for GBS---then allowing them to interfere in an optical network of beamsplitters and phaseshifters described by a unitary matrix $U$, and then measuring the output distribution of the photons (Fig.~\ref{fig:gbs}).

\begin{figure}[htb]
\centering
\includegraphics{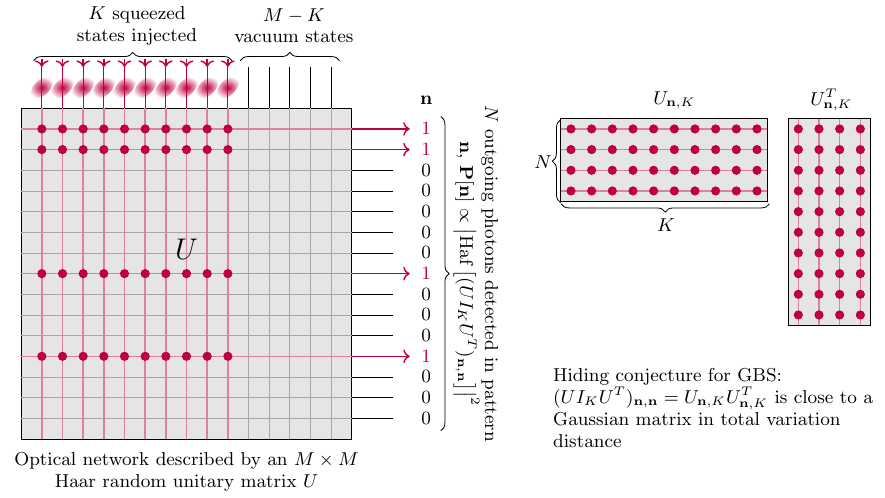}
\caption{A schematic of the Gaussian boson sampling (GBS) setup.  An initial Gaussian state with $K$ single-mode squeezed photon states is prepared. It then propagates through a linear-optical network described by an $M\times M$ unitary matrix $U$. The output modes are then measured in the photon number basis. We restrict to collision-free outputs (i.e., there is either 0 or 1 photon detected in each output channel). The Gaussian state does not have a definite photon number, hence the output has a variable number $N=\sum_{j=1}^M\mathbf n_j$ of outgoing photons, with the probability of a particular output distribution $\mathbf n$ given by \eqref{eqn:output}. This paper proves the hiding conjecture for Gaussian boson sampling with all input modes squeezed, i.e. $K=M$. This is a key part of the GBS hardness argument for this setting under other standard assumptions (which are analogous to assumptions used in the hardness argument for boson sampling).}
\label{fig:gbs}
\end{figure}

Compared to other proposals for demonstrating quantum advantage, GBS has already been realized in large-scale experiments as demonstrated in \cite{zhong2020quantum, Zhong2021, madsen2022, deng2023gaussian, liu2026gaussian}.
The proposed advantage from GBS is rooted in the classical difficulty of calculating the \emph{hafnian} of a large matrix, which is defined for a symmetric $2n\times2n$ matrix $A$ as
\begin{align*}
\Haf(A)&=\sum_{\pi\in\mathcal P_2(2n)}\prod_{\{i,j\}\in\pi}A_{ij},
\end{align*}
where $\mathcal P_2(2n)$ denotes all pairings of $\{1,\ldots,2n\}$, i.e. all possible partitions of $\{1,\ldots,2n\}$ into subsets of size two. 
As an example, when $A$ is the adjacency matrix of a graph, the hafnian of $A$ gives the number of perfect matchings of the graph. 
The hafnian generalizes the permanent of a matrix, and is related to the Pfaffian of a matrix in the same way that the permanent is related to the determinant.
However, unlike the determinant and Pfaffian, which can be calculated efficiently, calculating permanents and hafnians is known to be \sharpp-hard \cite{Valiant1979,Aaronson2011}.

In Gaussian boson sampling, the hafnian appears naturally in the probabilities for the output distribution of the photons. 
The initial state for GBS is a Gaussian state with $1\le K\le M$ single-mode squeezed states \cite{LundGBS2014,RahimiKeshariGBS2015,HamiltonGBS2017,KruseGBS2019} and $M-K$ vacuum states.
After allowing the photons to interfere in the $M$-mode interferometer described by a unitary matrix $U$, one can measure the output distribution in the photon number basis and obtain a sample $\mathbf n\in\{0,1,\ldots\}^M$ describing the photon counts for each of the $M$ output modes. 
For investigating classical hardness of simulating GBS, one often considers $K$ single-mode squeezed states all with the same squeezing parameter $r$, and only \emph{collision-free} outputs, meaning outputs $\mathbf n\in\{0,1\}^M$, which can be obtained by taking $r$ sufficiently small \cite{HamiltonGBS2017,KruseGBS2019,deshpande2022quantum}. Recent work \cite{grier2022complexity,bouland2023complexity} has also addressed the \emph{saturated} regime in boson sampling with many collisions, which is particularly relevant for experimental realizations.

In the collision-free case with an initial state of $K$ squeezed modes followed by $M-K$ vacuum modes, the probability of observing a GBS output photon distribution $\mathbf n\in \{0,1\}^M$ is given by \cite{HamiltonGBS2017,KruseGBS2019,deshpande2022quantum}
\begin{align}\label{eqn:output}
\P[\mathbf n]&=\frac{\tanh(r)^N}{\cosh(r)^K}\left|\Haf\left[\left(UI_KU^T\right)_{\mathbf{n},\mathbf{n}}\right]\right|^2,
\end{align}
where $U$ is the $M\times M$ unitary matrix describing the optical network, $N=\sum_{j=1}^M\mathbf n_j$, $I_K=\mathbf{1}_K\oplus\mathbf{0}_{M-K}$ is a diagonal matrix, and $(UI_KU^T)_{\mathbf n,\mathbf n}$ denotes the submatrix of $UI_KU^T$ formed by taking all the rows and columns $j$ where $\mathbf n_j=1$.
Since the squeezed states always have an even number $N$ of photons, $(UI_KU^T)_{\mathbf n,\mathbf n}$ has even dimensions which is compatible with the definition of the hafnian.
The submatrix $(UI_KU^T)_{\mathbf n,\mathbf n}$ can also be written as the symmetric product $U_{\mathbf n,K}U_{\mathbf n,K}^T$, where $U_{\mathbf n,K}$ is the $N\times K$ submatrix of $U$ formed by taking all the rows where $\mathbf n_j=1$ and the first $K$ columns (Figure~\ref{fig:gbs}).

A Gaussian boson sampler thus generates random samples $\mathbf n$ according to the distribution \eqref{eqn:output}, which involves probabilities proportional to the (modulus squared of the) hafnian of submatrices of $UI_KU^T$.
To argue that sampling from this distribution should be \emph{classically} hard, one relates the sampling problem to hardness of approximating the output probabilities, using Stockmeyer's algorithm \cite{stockmeyer1985approximation,aa,deshpande2022quantum,hangleiter2023computational}.
It is known that approximating the hafnian to within a certain additive or multiplicative error is \sharpp-hard in the worst case \cite{aa,deshpande2022quantum,hangleiter2023computational}, and it is conjectured that this hardness also holds for ensembles of complex matrices in the average case \cite{HamiltonGBS2017,KruseGBS2019,deshpande2022quantum,hangleiter2023computational}.
This average-case hardness is needed in the usual arguments to consider hardness of sampling from \emph{approximate} GBS, which allows for some error in the sampling distribution due to e.g. noise, photon loss, etc.
We note that assuming an average-case hardness conjecture like the above is typical---other random sampling schemes proposed to show quantum advantage \cite{hangleiter2023computational}, including the boson sampling protocol of \cite{aa}, also require these types of complexity theoretic conjectures, which remain major open problems in the area \cite{hangleiter2023computational}.

However, in addition to the typical complexity theory hardness conjecture, GBS has another conjecture, called the \emph{hiding conjecture} \cite{deshpande2022quantum,hangleiter2023computational}, which plays a key role in the above-mentioned argument for classical hardness. 
Roughly speaking, if we want to approximate the hafnian of a random matrix $Z$---which is conjectured to be a classically hard problem---using a Gaussian boson sampler or oracle, we first need to embed, or ``hide'', $Z$ inside a matrix of the form $UI_KU^T$, so that $|\Haf(Z)|^2$ will appear as a probability in the output distribution of GBS.
Due to the error allowed in approximate GBS, we need to be able to hide $Z$ in a strong sense so that an approximate sampler could not identify where we hid $Z$ and adversarially corrupt the probabilities corresponding to that particular submatrix.
The hiding conjecture asserts that this is possible, in a strong sense, for complex Gaussian-type $Z$ and a larger matrix $UI_KU^T$ for $U$ a Haar-random unitary matrix.
This is a key step in the usual argument for quantum advantage of Gaussian boson sampling. For further details, see \cite{aa,KruseGBS2019,deshpande2022quantum,hangleiter2023computational}.

As we explain in Section~\ref{subsec:fock}, the hiding conjecture for GBS in the experimentally relevant regime, where $K\propto M$ \cite{zhong2020quantum,Zhong2021,madsen2022,deng2023gaussian}, does not follow from the hiding property for Fock boson sampling from \cite{aa}. 
Due to this, the hardness argument for GBS experiments has so far lagged behind that of Fock boson sampling \cite{deshpande2022quantum}.
We also note that, partially due to the issue of the hiding conjecture for GBS, several alternate GBS schemes have been proposed, such as in \cite[{\S}VII]{KruseGBS2019} and the Bipartite Gaussian Boson Sampling of \cite{grier2022complexity}, which involve tuning the GBS device parameters in a specific way to avoid invoking the hiding conjecture.

To state the hiding conjecture more formally, let $\CN(0,\sigma^2)$ denote the complex Gaussian distribution with independent real and imaginary parts, each with mean zero and variance $\sigma^2/2$.
For the matrix $Z$ whose hafnian we want to approximate, we consider two related distributions:
\begin{itemize}
\item $\gsym$: The ensemble of $\q\times \q$ symmetric random matrices $\Gqq$ with $\CN(0,2)$ diagonal entries and $\CN(0,1)$ off-diagonal entries, with all entries independent modulo the symmetry requirement. 
\item $\ggtnk$: The ensemble of $N\times N$ matrices $GG^T$ where $G$ is an $N\times K$ matrix of i.i.d. $\CN(0,1/\sqrt{K})$ entries. The normalization of $G$ is chosen so that the entries of $GG^T$ have a nondegenerate limiting distribution as $K\to\infty$.
\end{itemize}
Note that $\ggtnk$ is \emph{not} a complex Wishart ensemble, as $G^T$ denotes the transpose, not the conjugate transpose.
Both ensembles $\gsym$ and $\ggtnk$ consist of symmetric complex matrices, and in certain regimes behave similarly (Corollary~\ref{cor:ggt}). We will reserve the notation $\Gqq$ for a matrix drawn from $\gsym$, and will use $GG^T$ to denote a matrix drawn from $\ggtnk$. 

Recall we are interested in submatrices $(UI_KU^T)_{\mathbf n,\mathbf n}=U_{\mathbf n,K}U_{\mathbf n,K}^T$, where $U_{\mathbf n,K}$ is the $N\times K$ submatrix of $U$ formed by taking all the rows where $\mathbf n_j=1$ and the first $K$ columns of a Haar unitary matrix $U$. Because the Haar measure is invariant under permutations of rows or columns, we may consider, without loss of generality, the distribution of the top left $N\times K$ submatrix of $U$, which we will denote by $U_{NK}$.
A precise form of the hiding conjecture for GBS, adapted from Conjecture 6.2 of \cite{deshpande2022quantum}, can then be stated as follows. 
While this formulation differs slightly from that of \cite{deshpande2022quantum}, we discuss afterwards how the changes are equivalent or sufficient.
\begin{conj}[hiding in GBS]\label{conj:hiding}
Let $U_{NK}$ be the top left $N\times K$ submatrix of an $M\times M$ Haar random unitary matrix, and let $Z=Z_{N,K}$ be a matrix with distribution given by either $\gsym$ or $\ggtnk$.
Then for $N\le K\le M$ and one of the distributions of $Z$, there exist polynomials $p,r$ such that for any $\delta>0$ and $M\ge p(N)/r(\delta)$, 
\begin{align}
d_\TV(MK^{-1/2}U_{NK}U_{NK}^T,Z)=O(\delta), \label{eqn:dtv}
\end{align}
where $d_\TV(\cdot,\cdot)$ denotes the total variation distance as defined below in \eqref{eqn:tv}.
\end{conj}

In this paper, we prove 
\begin{thm}\label{thm:hiding-conj}
Conjecture~\ref{conj:hiding} holds for $K=M$, with $Z$ drawn from $\gsym$.
\end{thm}
The case where all input modes $K=M$ are squeezed has been recently experimentally realized in \cite{madsen2022}. Theorem~\ref{thm:hiding-conj} gives the first rigorous proof of the hiding property for GBS in the experimentally relevant regime $K\propto M$, where in particular the hiding property cannot follow from hiding in Fock boson sampling.
We obtain Theorem~\ref{thm:hiding-conj} for the conjectured maximal size of $N=o(\sqrt{M})$.

\begin{rmk}\label{rmk:11}
\begin{enumerate}[(i)]
\item We prove Conjecture~\ref{conj:hiding} with $Z$ distributed as $\Gqq\dsim\gsym$, rather than with $GG^T\dsim\ggtnk$ which is how the conjecture is typically stated such as in \cite[Conjecture 6.2]{deshpande2022quantum}.
We argue this is actually an improvement: the matrix $\Gqq$ has independent Gaussian entries (modulo symmetry), making it in general much easier to work with than the $K$-dependent and correlated matrix $GG^T$. 
The resulting classical hardness argument then relies on conjectured hardness of approximating $\Haf(\Gqq)$, rather than of $\Haf(GG^T)$. Due to independent entries of $\Gqq$, we expect it should be much easier to study $\Haf(\Gqq)$ than $\Haf(GG^T)$.

\item Even if one preferred $GG^T$, we expect $\Haf(\Gqq)$ and $\Haf(GG^T)$ to probably be a similar hardness to approximate, since we have:
\begin{cor}\label{cor:ggt}
Let $N=N_K=o(\sqrt{K}/\log K)$. Then there are joint distributions of $(\Gqq,GG^T)$ such that $\Gqq\dsim\gsym$, $GG^T\dsim\ggtnk$, and
\begin{align}
\|\Gqq-GG^T\|_\infty\equiv\max_{1\le i,j\le N}\left|\Gqq_{ij} - (GG^T)_{ij}\right|\to0,
\end{align}
in probability as $K\to\infty$.
\end{cor}
The convergence in Corollary~\ref{cor:ggt} implies the entries of $\Gqq$ and $GG^T$ are very similar in the regime $N=o(\sqrt{K}/\log K)$. 
Since the hafnian is a direct polynomial function of the matrix entries, this entrywise closeness suggests that the problems of estimating $\Haf(\Gqq)$ and estimating $\Haf(GG^T)$ should be essentially the same in this regime. We note however that this corollary does not rigorously imply any hardness result.

We expect that for sufficiently small $N$ compared to $K$, that $\Gqq$ and $GG^T$ are close in total variation distance\footnote{Note however that even total variation convergence does \emph{not} imply convergence of moments of the hafnian.
Since this may have implications for anticoncentration, we note that one can check $\E|\Haf(\Gqq)|^2$ and $\E|\Haf(GG^T)|^2$, for $N=o(\sqrt{K})$, still have the same asymptotic behavior $(N-1)!!\sim\sqrt{2}N^{N/2}e^{-N/2}$ as $N\to\infty$, using the evaluation in \cite{Ehrenberg2025transition} for the latter moment.
} as $N,K\to\infty$ as well. This occurs for example for real Wishart matrices and the Gaussian orthogonal ensemble (GOE) \cite{RaczRichey2019}. 

\item The other change in Conjecture~\ref{conj:hiding} from \cite{deshpande2022quantum} is the allowance of a polynomial $r(\delta)$ instead of just $\delta$. But as noted in \cite{aa}, it only matters that we can take $M$ to be a polynomial in $N$ and $1/\delta$, so there is no loss in allowing for example $r(\delta)=\delta^2$ instead.
\end{enumerate}
\end{rmk}

\subsection{Comparison with hiding in Fock boson sampling}\label{subsec:fock}
In this section, we explain why the hiding property \cite[Theorem 5.1]{aa} for Fock boson sampling is not sufficient for Conjecture~\ref{conj:hiding} or Theorem~\ref{thm:hiding-conj}.
The hiding theorem for Fock boson sampling \cite{aa}, which we restate precisely in Theorem~\ref{thm:fhiding}, states that a $o(M^{1/5})\times o(M^{1/5})$ submatrix of an $M\times M$ Haar unitary matrix $U$ is close, in total variation distance, to a matrix of i.i.d. standard complex Gaussian random variables. 
This can be extended up to the conjectured maximal size $o(\sqrt{M})\times o(\sqrt{M})$ submatrix as we discuss in Section~\ref{subsec:other} and Appendix~\ref{app:sparse}.
This Fock boson sampling hiding property then does imply Conjecture~\ref{conj:hiding} in a certain sparse regime \cite{deshpande2022quantum}, $N=K=o(\sqrt{M})$. Yet this is not the regime of interest for large-scale experimental realizations of GBS, which have $K=cM$ \cite{zhong2020quantum,Zhong2021,madsen2022,deng2023gaussian}. (If $K$ is taken too small, then it is hard to obtain a large enough photon number $N$. Moreover, large $K\gg N^2$ is favorable for the anticoncentration results of \cite{Ehrenberg2025transition,Ehrenberg2025second}, which also provide evidence for hardness of approximation.) 
By taking $K=cM$, the hiding property for GBS now involves a large rectangular $N\times K$ submatrix $U_{NK}$, which becomes too large to approximate by Gaussians in total variation distance (cf. Fig.~\ref{fig:submat} and Section~\ref{subsec:other}).
The experimental realizations of GBS \cite{zhong2020quantum,Zhong2021,madsen2022,deng2023gaussian}, with $K=cM$, therefore fall outside of the ``sparse regime.'' 
In fact, for $K=M$, the submatrix $U_{NK}$ itself is maximally \emph{far} from Gaussian in total variation distance due to the row normalization requirement: for any $N\ge1$ one has $d_\TV(\sqrt{M}U_{NM},G)=1$, for $G$ an $N\times M$ matrix of i.i.d. standard complex Gaussians. Therefore we see the hiding property for $U_{NK}U_{NK}^T$ proved in Theorem~\ref{thm:hiding-conj} cannot follow from a hiding property for $U_{NK}$.

\begin{figure}[htb]
\centering
\includegraphics{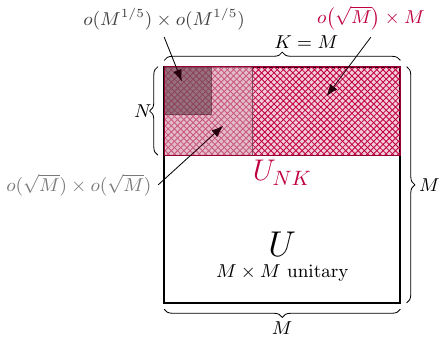}





\caption{Comparison of the submatrix sizes considered for Fock boson sampling (gray) and Gaussian boson sampling in Theorem~\ref{thm:hiding-conj} (cross-hatched, purple). 
Hiding for Fock boson sampling involves submatrices of size $o(M^{1/5})\times o(M^{1/5})$ (or up to $o(\sqrt{M})\times o(\sqrt{M})$), which look Gaussian by \cite[Theorem 5.1]{aa} (and Theorem~\ref{thm:fock-hiding}). However, hiding for Gaussian boson sampling in Theorem~\ref{thm:hiding-conj} involves the symmetric product of a much larger $N\times K=o(\sqrt{M})\times M$ rectangular submatrix $U_{NK}$, which itself \emph{cannot} be close to Gaussian in TV distance---due to the row normalization requirement, when $K=M$, then $d_\TV(\sqrt{M}U_{NK},G)=1$ for any $N\ge1$. Theorem~\ref{thm:hiding-conj} says that even though $U_{NK}$ itself is far from Gaussian, the symmetric product $U_{NK}U_{NK}^T$ nevertheless looks like a Gaussian matrix $\Gqq\dsim\gsym$.
}\label{fig:submat}
\end{figure}

\subsection{Precise formulation of main result}

We now give a more precise form of our main result Theorem~\ref{thm:hiding-conj}.
First recall the total variation (TV) distance between two probability measures $\mu$ and $\nu$ on a measure space $(E,\mathcal E)$ is defined as
\begin{align}\label{eqn:tv}
d_\TV(\mu,\nu)&\equiv\sup_{A\in\mathcal E}|\mu(A)-\nu(A)|.
\end{align}
When $\mu$ and $\nu$ have densities $f$ and $g$ respectively with respect to a measure $d\lambda$ on $E$, then
\begin{align}\label{eqn:tv2}
d_\TV(\mu,\nu)&=\frac{1}{2}\int_E|f(x)-g(x)|\,d\lambda(x).
\end{align}
For random variables $Z$ and $W$, we will write $d_\TV(Z,W)$ to mean the total variation distance between their distributions.

For Gaussian boson sampling, when $K=M$, the matrix $U_{NK}U_{NK}^T$ is distributed as a submatrix of a circular orthogonal ensemble \cite{Dyson1962} (COE) random matrix. 
The COE ensemble consists of symmetric unitary matrices, with a COE matrix being distributed as $UU^T$ where $U$ is a standard Haar-random unitary matrix. 
The matrix $U_{NM}U_{NM}^T$ is the upper left $N\times N$ submatrix of $UU^T$, and thus $U_{NM}U_{NM}^T$ is distributed as a submatrix of a COE random matrix.
Therefore the hiding conjecture when $K=M$ reduces to proving that a submatrix of a COE matrix is close to a Gaussian matrix in total variation distance. We thus restate Theorem~\ref{thm:hiding-conj} more precisely as follows:

\begin{thm}[hiding in GBS with $K=M$]\label{thm:hiding}
Let $\Uqq$ be the upper left $\q\times\q$ submatrix of an $M\times M$ COE matrix $U$, and let $\Gqq\dsim\gsym$. 
Then for $\q=O(\sqrt{M})$, as $M\to\infty$,
\begin{align}
d_\TV(\sqrt{M}\Uqq,\Gqq)&\le O(\q/\sqrt{M}),
\end{align}
which is $o(1)$ for $N=o(\sqrt{M})$.
\end{thm}
The main difficulty is that the convergence is in the total variation distance. Closeness of $\sqrt{M}\Uqq$ to complex Gaussian in weaker senses, such as in distribution, has been long known \cite{FriedmanMello1985,CollinsStolz2008,Jiang2009}. In fact \cite{FriedmanMello1985} already noted the closeness of the density of $\Uqq$ to that of a Gaussian distribution for sufficiently large $M$.
However, these weaker forms of convergence are not sufficient for the hiding property and subsequent hardness argument \cite{aa,hangleiter2023computational}. Moreover, as shown in \cite{Jiang2006}, closeness of a submatrix in total variation distance can be very different than in other senses of convergence. In particular, for a square submatrix of a Haar orthogonal matrix, one can approximate only up to a size $o(\sqrt{M})\times o(\sqrt{M})$ submatrix in total variation distance, but one can approximate an entire $M\times o(\frac{M}{\log M})$ submatrix in the weak or maximum entrywise distance \cite[Table 1]{JiangMa2019}.

The proof of Theorem~\ref{thm:hiding} follows that of \cite[Theorem 1]{JiangMa2019}, which involves bounding the Kullback--Leibler (KL) divergence of the two distributions. One notable difference is that we do not determine the normalizing constant in the formula for the density of $\sqrt{M}\Uqq$ in Theorem~\ref{lem:submatrix-density}.
The normalizing constant appears complicated to calculate, and by skipping it we also skip some detailed asymptotic analysis, at the cost of adding Proposition~\ref{prop:zeta} to handle the unknown normalizing constant.
For this, we combine the method of \cite{aa}, which also did not determine the normalizing constant, but as a byproduct had to assume a smaller size submatrix, with KL divergence-like estimates to prove the TV distance convergence up to the expected maximal size $o(M^{1/2})$ of the square submatrix $\Uqq$.
Being able to skip calculating the normalizing constant without losing out on the submatrix size, may be useful for proving convergence for other submatrix distributions as well.

\begin{rmk}
After uploading the original version of this manuscript, we found the integral for the normalizing constant for the density in Theorem~\ref{lem:submatrix-density} is calculated in the 1963 book by Hua \cite[Theorem 2.3.1]{Hua1963}, as
\begin{align}
c_{M,N}'&=\frac{(M-2N+1)(M-2N+3)\cdots(M-1)}{2^N\pi^{N(N+1)/2}}\frac{\Gamma(M-N+1)\Gamma(M-N+2)\cdots\Gamma(M-1)}{\Gamma(M-2N+2)\Gamma(M-2N+4)\cdots\Gamma(M-2)}.
\end{align}
Therefore, Proposition~\ref{prop:zeta} could be replaced with an asymptotic expansion of $c_{M,N}'$ similar to \cite[Lemma 2.7]{JiangMa2019}.
However, we keep the proposition and proof with the undetermined normalizing constant since it may be useful for other random matrix ensembles, as it appears typical (and easier) to derive submatrix densities without the precise normalizing constant, as done in \cite{FriedmanMello1985,Eaton1983book,Beenakker1997random,Forrester-book,collins2003thesis}.
Additionally, we avoid some computational work involving the normalizing constant, in exchange for the different types of estimates in Proposition~\ref{prop:zeta}.
\end{rmk}

Returning to the argument for classical hardness, we also need a multiplicative estimate on the density functions to apply the instance generating method of \cite{aa}:
\begin{thm}[hiding v2]\label{thm:instance}
Let $\Uqq$ and $\Gqq$ be as in Theorem~\ref{thm:hiding}, and let $f$ be the density of $\sqrt{M}\Uqq$ and $g$ the density of $\Gqq$. Then for $\q=o(M^{1/3})$ and any $\q\times\q$ symmetric matrix $Z$,
\begin{align}
f(Z)&\le (1+O(\q^3/M))g(Z).
\end{align}
\end{thm}
Conveniently, this follows from the proof of Theorem~\ref{thm:hiding} in the same way that the boson sampling analogue follows from the boson sampling hiding property \cite[Theorems 5.1, 5.2]{aa}.
This type of theorem is used to generate a random unitary given a particular Gaussian matrix that must be hidden as a submatrix \cite{aa}. For Gaussian boson sampling with the hiding property involving $\Gqq\dsim\gsym$, one can use Theorem~\ref{thm:instance} and the argument in \cite[\S5.2]{aa} to generate a random COE matrix $B$ containing $\Gqq$ (properly rescaled) as a randomly placed submatrix. 
Given a COE matrix $B$, singular value decomposition or Takagi factorization can produce a unitary matrix $U$ such that $B=UU^T$.

\subsection{Other numbers of squeezed states}\label{subsec:other}

In this section we discuss the hiding conjecture for $K<M$. 
Previously, the only case where Conjecture~\ref{conj:hiding} was rigorously known to hold was in the sparse case $N=K=o(M^{1/5})$, where the hiding property for GBS follows from the hiding property for conventional boson sampling proved in \cite[Theorem 5.1]{aa}:
\begin{thm}[{\cite[Theorem 5.1]{aa}}]\label{thm:fhiding}
Let $U_{NN}$ be the top left $N\times N$ submatrix of an $M\times M$ Haar unitary matrix $U$, and let $G_N$ be an $N\times N$ matrix of i.i.d. standard complex Gaussians.
Let $M\ge\frac{N^5}{\delta}\log^2\frac{N}{\delta}$ for any $\delta>0$. Then $d_\TV(\sqrt{M}U_{NN},G_N)=O(\delta)$.
\end{thm}
Note the hiding property for conventional boson sampling involves the square submatrix $U_{NN}$ of a Haar unitary $U$, while the hiding property for GBS involves the symmetric product $U_{NK}U_{NK}^T$ for a rectangular $N\times K$ submatrix $U_{NK}$.
In the sparse regime for GBS, when $N=K=o(M^{1/5})$, the hiding property for GBS follows from Theorem~\ref{thm:fhiding}, since $U_{NK}$ itself looks like a Gaussian $G$ in total variation distance, and so $U_{NK}U_{NK}^T$ will also be close to $GG^T$ in total variation distance. 
Actually, applying \cite{Jiang2009jacobi,JiangMa2019}, the hiding property in the sparse regime can be extended to its maximum extent, $K=O(M^{1-\varepsilon})$, for any $0<\varepsilon<1$, albeit with requiring $N=o(M^\varepsilon)$.
We write the details in Appendix~\ref{app:sparse}, which just consists of applying the proof method of \cite[Theorem 1(i)]{JiangMa2019} to the unitary submatrix density in \cite[Prop. 2.1]{Jiang2009jacobi}. As a by-product the method of \cite{Jiang2009jacobi,JiangMa2019} also improves Theorem~\ref{thm:fhiding} to the expected maximal size $N=o(\sqrt{M})$ for hiding in conventional boson sampling, with a quantitative rate of convergence, as conjectured in \cite{aa} (Theorem~\ref{thm:fock-hiding}/\ref{thm:sparse-hiding}).
\begin{thm}[hiding in GBS for $NK=o(M)$]\label{thm:sparse-hiding0}
Let $1\le N\le K<M$ and $NK=o(M)$. Then for $G_{NK}$ an $N\times K$ matrix of i.i.d. standard complex Gaussians,
\begin{align}
d_\TV(MU_{NK}U_{NK}^T,G_{NK}G_{NK}^T)&\le O\left(\sqrt{\frac{NK}{M}}\right).
\end{align}
In particular, if $K=O(M^{1-\varepsilon})$ and $N=O(M^{\varepsilon/2})$ for some $\varepsilon>0$, then the TVD bound is $O(M^{-\varepsilon/4})$.
\end{thm}

\begin{thm}[hiding in Fock boson sampling for $N=o(\sqrt{M})$]\label{thm:fock-hiding}
Total variation distance hiding in Fock boson sampling holds for submatrices of size $N=o(\sqrt{M})$: for $G_N$ an $N\times N$ matrix of i.i.d. standard complex Gaussians and $N=o(\sqrt{M})$,
\begin{align}
d_\TV(\sqrt{M}U_{NN},G_{NN})\le O\left(\frac{N}{\sqrt{M}}\right).
\end{align}
\end{thm}

However, as soon as we consider $K=cM$, even if $N$ is as small as $N=1$, then the $N\times K$ submatrix $U_{NK}$ should be \emph{too large} to be close to a Gaussian $G$ in TV distance---this is the case for submatrices of Haar orthogonal matrices \cite{JiangMa2019}, and the same proof method should carry over to unitary submatrices as well.
Assuming this, there is a transition for whether an $N\times K$ submatrix $U_{NK}$ looks Gaussian in TV distance as $M\to\infty$: for $NK=o(M)$ the TV distance goes to zero and $U_{NK}$ looks Gaussian, but for $NK=\Omega(M)$, the submatrix $U_{NK}$ does \emph{not} look Gaussian, i.e. the TV distance is bounded away from zero as $M\to\infty$ \cite{JiangMa2019}.
As previously noted, we can see that when $K=M$, in which case the rows of $U_{NK}$ are normalized, then $d_\TV(\sqrt{M}U_{NK},G)=1$, which is the maximum possible value for TV distance.

So for the experimental regime of $K=cM$, the hiding property for GBS involves $U_{NK}U_{NK}^T$ for $U_{NK}$ an $N\times K$ submatrix which does \emph{not} look Gaussian in TV distance. Nevertheless, Conjecture~\ref{conj:hiding} asserts the end product $U_{NK}U_{NK}^T$ may still look either like a symmetric Gaussian $\Gqq$, or like $GG^T$ for $G$ i.i.d. Gaussian.
Our main result Theorem~\ref{thm:hiding-conj} or \ref{thm:hiding} proves this for $K=M$.
Intuitively, the hiding property should be \emph{least likely} to hold when $K=M$, because then the individual $U_{NK}$ and $G$ are most different, and $U_{NK}$ has strong relations between entries due to the normalization and orthogonality requirements on its rows. 
Since the hiding property holds for $K=M$, then it seems intuitively plausible it should also hold for smaller $K=cM$, $K<M$, where the entries of $U_{NK}$ look more independent since they do not see the full normalization and orthogonality requirements.
Our proof for hiding in GBS for $K=M$ however makes use of the explicit form of the density for $U_{NK}U_{NK}^T$ (just like other related total variation distance proofs), and so it unfortunately does not imply hiding for $K=cM$, $K<M$, despite the intuition that the setting with $K<M$ should be in some sense ``easier''.

Finally, we note that in \cite{HamiltonGBS2017,deshpande2022quantum}, the case $K=M$ is also discussed in relation to \cite{Jiang2009}, but the convergence in \cite{Jiang2009} is in terms of the weak or maximum entrywise distance. 
The proof using Gram--Schmidt orthonormalization in \cite{Jiang2009} also applies to $U_{NK}U_{NK}^T$ with $K<M$, so $MK^{-1/2}U_{NK}U_{NK}^T$ also looks like $GG^T\dsim\ggtnk$ in the weaker maximum entrywise distance for $N=o(\sqrt{M}/\log M)$. But as previously discussed, this weaker distance can behave very differently from the total variation distance, and is not sufficient for the hiding conjecture or subsequent hardness application.

\subsection{Outline and notation}

The rest of the paper is organized as follows.
The main Theorem~\ref{thm:hiding} is proved through Sections~\ref{sec:coe} and \ref{sec:proof}, with some lemmas and their proofs deferred to Section~\ref{sec:lemmas}.
Section~\ref{sec:instance} gives the proofs of Corollary~\ref{cor:ggt} and Theorem~\ref{thm:instance}, and Appendix~\ref{app:sparse} gives the proof of Theorem~\ref{thm:sparse-hiding0}.

For notation, generic constants denoted by $C$ or $c$ may change from line to line. 
For random variables $X,Y$ we write $X\deq Y$ to mean $X$ and $Y$ have the same distribution, and for $\mu$ a probability distribution we write $X\dsim\mu$ to mean $X$ has distribution given by $\mu$.

\section{COE preliminaries}\label{sec:coe}

We need a few properties of submatrices of COE matrices. 
The behavior of these submatrices has been of interest in quantum conductance and transport problems \cite{Beenakker1997random,Forrester-book}.
The key property we need is the explicit formula for the density of a submatrix of a COE matrix derived in \cite{FriedmanMello1985}. 
This is what makes proving convergence in total variation distance possible.
\begin{thm}[submatrix density, \cite{FriedmanMello1985}]\label{lem:submatrix-density}
Let $\Uqq$ be the upper left $\q\times\q$ submatrix  of an $M\times M$ COE matrix $U$. Then for $N\le M/2$, $\Uqq$ has a probability density function $f$ on the space of symmetric complex matrices given by
\begin{align}\label{eqn:submat-density}
f(Z)=c_{M,\q}'\det(I_\q-Z^\dagger Z)^{\frac{M-2\q-1}{2}}\iz(\lambda_\mathrm{max}(Z^\dagger Z)<1),
\end{align}
for a normalization constant $c_{M,\q}'$, and where $\iz(\cdot)$ denotes an indicator function and $\lambda_\mathrm{max}(Z^\dagger Z)$ the maximum eigenvalue of $Z^\dagger Z$.
\end{thm}
This gives the distribution of $\Uqq$ as $f(Z)\prod_{i\le j}dZ_{ij}$ on $\C^{\q(\q+1)/2}$.
Note that the density \eqref{eqn:submat-density} is essentially the same form as the density for the submatrix of a Haar-distributed real orthogonal matrix \cite{DEL1992,Eaton1989,Meckes-book}, although the integration space (and dimension) are different.
From \eqref{eqn:submat-density}, we can see the standard intuitive reason behind closeness of $\Uqq$ to Gaussian. Considering $\sqrt{M}\Uqq$, the density function is rescaled to be 
\begin{align}\label{eqn:approx}
f(Z)\propto \det\left(I_\q-\frac{Z^\dagger Z}{M}\right)^{\frac{M-2\q-1}{2}}=\prod_{j=1}^\q\left(1-\frac{\lambda_j(Z^\dagger Z)}{M}\right)^{\frac{M-2\q-1}{2}}\approx\prod_{j=1}^\q e^{-\lambda_j(Z^\dagger Z)/2}=e^{-\Tr(Z^\dagger Z)/2},
\end{align}
provided the eigenvalues $\lambda_j(Z^\dagger Z)$ of $Z^\dagger Z$ are sufficiently small for the exponential approximation. We will of course need much more precise estimates to show the TV distance \eqref{eqn:tv2} is small, which are slightly complicated by not knowing the normalization constant on the density.
We note that, despite the intuitive closeness of the density functions such as in \eqref{eqn:approx}, studying the TV distance in the related case of submatrices of Haar-distributed real orthogonal matrices still took many years to finally be fully resolved \cite{DiaconisFreedman1987,DEL1992,Jiang2006,JiangMa2019,Stewart2020}. Moreover, there are several different notions, besides total variation distance, in which the matrices can be close, and the behavior with respect to these other distances can greatly differ \cite{Meckes-book}, \cite[Table 1]{JiangMa2019}.
Therefore we will need much more careful analysis than \eqref{eqn:approx} to prove Theorem~\ref{thm:hiding}.

We also need the following eigenvalue estimates, which are the analogue of \cite[Lemmas 2.5, 2.6]{JiangMa2019} but for COE submatrices.
\begin{lem}[singular values]\label{lem:submatrix-sv}
Let $\Uqq$ be the upper left $\q\times\q$ submatrix  of an $M\times M$ COE matrix $U$, and let $\lambda_1,\ldots,\lambda_\q$ be the eigenvalues of $\Uqq^\dagger\Uqq$.
Then for $\q=O(\sqrt{M})$, as $M\to\infty$,
\begin{align}
\E \sum_{j=1}^\q \lambda_j &=\frac{\q^2+\q}{M}+O\left(\frac{\q^2}{M^2}\right),\label{eqn:lambda1}\\
\E \sum_{j=1}^\q \lambda_j^2&=\frac{2\q^3}{M^2}+O\left(\frac{\q^2}{M^2}\right),\label{eqn:lambda2}\\
\E \sum_{j=1}^\q \lambda_j^3&=O\left(\frac{\q^4}{M^3}\right).\label{eqn:lambda3}
\end{align}
\end{lem}
Since $\E\sum_{j=1}^\q \lambda_j^m=\E\Tr((\Uqq^\dagger\Uqq)^m)$, we can calculate all the above quantities using Weingarten calculus \cite{Collins2003,CollinsSniady2006}, writing a COE matrix as $U\overset{d}{=}WW^T$ for $W$ a Haar-distributed unitary, and identifying leading order terms. 
We postpone the proof of Lemma~\ref{lem:submatrix-sv} to Section~\ref{sec:coe-moments}.

\section{Proof of Theorem~\ref{thm:hiding}}\label{sec:proof}

There is a long history of proofs showing closeness of a Haar orthogonal or unitary submatrix to i.i.d. Gaussians in total variation distance, including \cite{DiaconisFreedman1987,DEL1992,Jiang2006,Jiang2009jacobi,aa,JiangMa2019,Stewart2020}, and concluding with the result for maximal sizes of rectangular submatrices of Haar orthogonal matrices \cite{JiangMa2019,Stewart2020}.
For the COE case here, we adapt the proof method of \cite{JiangMa2019}, which is for Haar orthogonal matrices, and gives in that case the maximal submatrix size and a quantitative error bound. Additionally, the proof method is relatively short and provides a bound on the Kullback--Leibler divergence (or relative entropy).
By Pinsker's inequality \cite{Pinsker1964}, the total variation distance can be bounded as
\begin{align}\label{eqn:pinsker}
d_\TV(\sqrt{M}\Uqq,\Gqq)&\le \sqrt{\frac{1}{2}\dkl(\sqrt{M}\Uqq\;||\;\Gqq)},
\end{align}
where $\dkl(\sqrt{M}\Uqq\;||\;\Gqq)$ is the Kullback--Leibler (KL) divergence (relative entropy) of the distributions of $\sqrt{M}\Uqq$ and $\Gqq$, 
\begin{align}\label{eqn:kl}
\dkl(\sqrt{M}\Uqq\;||\;\Gqq)&\equiv\int f(Z)\log\frac{f(Z)}{g(Z)}\,dZ,
\end{align}
where $f$ and $g$ are the density functions for $\sqrt{M}\Uqq$ and $\Gqq$ respectively. 
The density $f$ of $\sqrt{M}\Uqq$ is given by a rescaling of Theorem~\ref{lem:submatrix-density} as
\begin{align}\label{eqn:msubmat-density}
f(Z)&=c_{M,\q}\det\Big(I_\q-\frac{Z^\dagger Z}{M}\Big)^{\frac{M-2\q-1}{2}}\iz(\lambda_\mathrm{max}Z^\dagger Z<M),
\end{align}
for a normalization constant $c_{M,\q}$,
while the density $g$ is
\begin{align}
g(Z)&=\prod_{j=1}^N\frac{1}{2\pi}e^{-|Z_{jj}|^2/2}\prod_{i<j}\frac{1}{\pi}e^{-|Z_{ij}|^2}=\frac{1}{2^N\pi^{N(N+1)/2}}e^{-\Tr(Z^\dagger Z)/2}.
\end{align}
The KL divergence \eqref{eqn:kl} is well-defined since the distribution of $\sqrt{M}\Uqq$ is absolutely continuous with respect to that of $\Gqq$.
Due to \eqref{eqn:pinsker} it then suffices to bound the KL divergence.

As explained in the Introduction, we skip the exact calculation of the normalizing constant $c_{M,N}$. In fact \cite{aa} does not evaluate their normalizing constant either, although in both of these cases it turns out the normalizing constants can be found in the matrix integrals in the book by Hua \cite{Hua1963}.
Instead, we have the additional estimates of Proposition~\ref{prop:zeta}. 
The intuitive reason for why we do not need the constant is that if two density functions $f$ and $g$ are ``close'' without certain normalizing constants, then those normalizing constants must also be close since $f,g$ are both density functions normalized to have integral one.

To avoid $c_{M,\q}$, we first define a more convenient function $\tilde f$ with an explicit constant, and apply the argument involving the Kullback--Leibler divergence from \cite{JiangMa2019} to bound the expression in \eqref{eqn:pinsker} but with $\tilde f$ replacing $f$. To this end, let $\zeta=\zeta_{M,\q}$ be defined as in \eqref{eqn:zeta} below, and define $\tilde f:=\zeta f$. From \eqref{eqn:kl}, we can write
\begin{align}\label{eqn:dkl-split}
\dkl(\sqrt{M}\Uqq\;||\;\Gqq)&=\int f(Z)\log\left(\frac{\frac{1}{\zeta}\tilde f(Z)}{g(Z)}\right)\,dZ
=\E\left[\log\frac{\tilde f(\sqrt{M}\Uqq)}{g(\sqrt{M}\Uqq)}\right]+\log\frac{1}{\zeta}.
\end{align}
We will first show the expectation term is small, and then handle the normalizing constant term $\log(1/\zeta)$ in Proposition~\ref{prop:zeta}.

\begin{prop}\label{prop:exp}
Let $\Uqq$ be the upper left $\q\times\q$ submatrix of an $M\times M$ COE matrix, which has density $f=f_{M,\q}$ given in \eqref{eqn:msubmat-density}.
Let $g=g_\q$ be the density function of $\Gqq\sim\gsym$, and let $\tilde f:=\zeta f$ for 
\begin{align}\label{eqn:zeta}
\zeta=\zeta_{M,\q}:=\frac{\exp\left(-\frac{\q^3}{2M}\right)}{2^\q\pi^{\q(\q+1)/2}c_{M,\q}}.
\end{align}
Then for $\q=O(\sqrt{M})$, as $M\to\infty$, there is a constant $C$ so that
\begin{align}
\E\left[\log\frac{\tilde f(\sqrt{M}\Uqq)}{g(\sqrt{M}\Uqq)}\right]\le C\q^2/M.
\end{align}
\end{prop}
\begin{proof}
Note since the expectation we evaluate is no longer a KL divergence, it can be negative, but we only need an upper bound.
We follow the proof method of \cite[Theorem 1(i)]{JiangMa2019} with some simplifications since we consider square submatrices and have pre-chosen the constant $\zeta$ in \eqref{eqn:zeta} to be particularly convenient; in particular no asymptotic expansion of the normalization constant is needed.
The density for $\Gqq$ is $g(Z)=2^{-\q}\pi^{-\q(\q+1)/2}e^{-\Tr(Z^\dagger Z)/2}$.
In terms of the eigenvalues $\lambda_1,\ldots,\lambda_\q$ of the matrix $Z^\dagger Z$, we have for $\lambda_\mathrm{max}(Z^\dagger Z)<M$,
\begin{align*}
\log\frac{\tilde f(Z)}{g(Z)}&=
-\frac{\q^3}{2M}
+\log\left[\det\left(I_\q-\frac{Z^\dagger Z}{M}\right)^{\frac{M-2\q-1}{2}}\exp\left(\Tr(Z^\dagger Z)/2\right)\right]
\\
&=-\frac{\q^3}{2M}
+\frac{M-2\q-1}{2}\sum_{j=1}^\q\log\Big(1-\frac{\lambda_j(Z^\dagger Z)}{M}\Big)+\frac{1}{2}\sum_{j=1}^\q\lambda_j(Z^\dagger Z).
\numberthis\label{eqn:logratio}
\end{align*}
We can take the expectation value with respect to $\sqrt{M}\Uqq$, for which $\lambda_\mathrm{max}(M\Uqq^\dagger\Uqq)\le M$. Using the inequality $\log(1+x)\le x-\frac{x^2}{2}+\frac{x^3}{3}$ for any $x>-1$, and applying Lemma~\ref{lem:submatrix-sv} (recall we have rescaled $\lambda_j\mapsto M\lambda_j$), yields
\begin{align*}
\E \log\frac{\tilde f(\sqrt{M}\Uqq)}{g(\sqrt{M}\Uqq)}&\le -\frac{\q^3}{2M}
+\frac{(M-2\q-1)}{2}\E\sum_{j=1}^N\left[-\frac{\lambda_j}{M}-\frac{\lambda_j^2}{2M^2}-\frac{\lambda_j^3}{3M^3}\right]+\frac{\q^2+\q}{2}+O\left(\frac{\q^2}{M}\right)\\
&\le
-\frac{\q^3}{2M}
+\frac{(M-2\q-1)}{2}\left[-\frac{\q^2+\q}{M}-\frac{2\q^3}{2M^2}-O\Big(\frac{\q^4}{M^3}\Big)\right]+\frac{\q^2+\q}{2}+O\left(\frac{\q^2}{M}\right)\\
&=O\left(\frac{\q^2}{M}\right),\numberthis
\end{align*}
for $N=O(\sqrt{M})$.
\end{proof}

Next we handle the undetermined normalizing constant $c_{M,\q}$ by showing that the term $\log1/\zeta$ in \eqref{eqn:dkl-split} is small. 
\begin{prop}\label{prop:zeta}
Let $\zeta=\zeta_{M,\q}$ be defined as in \eqref{eqn:zeta}. Then there is $c_0>0$ so that for $N\le c_0\sqrt{M}$, as $M\to\infty$,
\begin{align}
|1-\zeta|&=O(\q^2/M),\quad
\text{and consequently }|\log\zeta|=O(\q^2/M).
\end{align}
\end{prop}

\begin{proof}
Recall that $\tilde f\equiv\zeta f$, where $f$ is the density function of $\sqrt{M}\Uqq$, and that $g$ denotes the density function of $\Gqq\sim\gsym$.
Because we want to compare to a KL divergence-like term, we use a different method than what was used in \cite{aa}.
We start by observing the function $x\log x-x+1$ is nonnegative for all $x\ge0$. Therefore
\begin{align}
0\le \int_{\C^{\q(\q+1)/2}}\left(\frac{\tilde f}{g}\log\frac{\tilde f}{g}-\frac{\tilde f}{g}+1\right)g\,dZ&=\zeta\E\left[\log\frac{\tilde f(\sqrt{M}\Uqq)}{g(\sqrt{M}\Uqq)}\right]-\zeta+1,
\end{align}
which with Proposition~\ref{prop:exp} implies
\begin{align}
\zeta&\le 1+O(\q^2/M).
\end{align}
To show a lower bound on $\zeta$, let $S:=\{\mathbf{Z}\text{ complex symmetric $\q\times\q$ matrix}: \lambda_\mathrm{max}(\mathbf{Z}^\dagger\mathbf{Z})\le M/2\}$, which is a (proper) subset of where $\tilde f>0$. Then we have
\begin{align}\label{eqn:zeta2}
0\le \int_{S}\left(\frac{g}{\tilde f}\log\frac{g}{\tilde f}-\frac{g}{\tilde f}+1\right)\tilde f\,dZ &=\E\left[\oneb_{S}(\Gqq)\log\frac{g(\Gqq)}{\tilde f(\Gqq)}\right]-\E\left[\oneb_{S}(\Gqq)\right]+\zeta\E[\oneb_S(\sqrt{M}\Uqq)].
\end{align}
We will show that the event $S$ occurs with high probability for both $\Gqq$ and $\sqrt{M}\Uqq$, so the right-most two expectations are both close to one. 
The remaining expectation will be similar to the one computed in Proposition~\ref{prop:exp}, but using the restriction $\oneb_S(\Gqq)$, so that $\lambda_j(\Gqq^\dagger\Gqq)< M/2$, to have a \emph{lower} bound on $\log(1-\lambda_j/M)$.
To show the above, we use some lemmas concerning the singular values of $\Gqq$ and $\Uqq$, which we state and prove in Section~\ref{subsec:g-sv}.
In particular, by Lemma~\ref{lem:g-maxsv}, letting $\lambda_1'\,\ldots,\lambda_N'$ denote the eigenvalues of $\Gqq^\dagger\Gqq$,
\begin{align}\label{eqn:maxsv}
\E\left[\oneb_{S}(\Gqq)\right]&=1-\P[\max(\lambda_1',\ldots,\lambda_\q')>M/2]\ge 1- e^{-c'M},
\end{align}
for sufficiently small $c_0$ in $N\le c_0\sqrt{M}$. Similarly, by Lemma~\ref{lem:u-maxsv}, for $\lambda_1,\ldots,\lambda_\q$ the eigenvalues of $M\Uqq^\dagger\Uqq$,
\begin{align}\label{eqn:amaxsv}
\E\left[\oneb_S(\sqrt{M}\Uqq)\right]&=1-\P[\max(\lambda_1,\ldots,\lambda_\q)>M/2] \ge 1-O(\q^2/M).
\end{align}
Finally, using Lemma~\ref{lem:gsv} in place of Lemma~\ref{lem:submatrix-sv}, we can upper bound $\E\left[\oneb_{S}(\Gqq)\log\frac{g(\Gqq)}{\tilde f(\Gqq)}\right]$ similarly as in the proof of Proposition~\ref{prop:exp}. Instead of \eqref{eqn:logratio}, we have the negation with an indicator function,
\begin{multline}
\oneb_S(Z)\log\frac{g(Z)}{\tilde f(Z)}
=\oneb_S(Z)\bigg[\frac{\q^3}{2M}
-\frac{M-2\q-1}{2}\sum_{j=1}^\q\log\Big(1-\frac{\lambda_j(Z^\dagger Z)}{M}\Big)-\frac{1}{2}\sum_{j=1}^\q\lambda_j(Z^\dagger Z)\bigg].\numberthis\label{eqn:logratio2}
\end{multline}
Due to the negative sign on the $\log(1-\lambda_j/M)$ terms, we need a \emph{lower} bound on $\log(1-x)$, for which we will use
\begin{align}
\log(1-x)&\ge -x-\frac{x^2}{2}-x^3,\quad \text{for }0\le x\le 1/2,
\end{align}
which can be verified by differentiation.
Because the set $S$ is defined so $0\le \frac{\lambda_j(Z^\dagger Z)}{M}\le 1/2$, we can use the above inequality to get the lower bound
\begin{align}\label{eqn:log-lower}
\E\sum_{j=1}^\q\log\Big(1-\frac{\lambda_j'(\Gqq^\dagger\Gqq)}{M}\Big)\oneb_S(\Gqq)&\ge \E\left[\sum_{j=1}^\q\left(-\frac{\lambda_j'}{M}-\frac{(\lambda_j')^2}{2M^2}-\frac{(\lambda_j')^3}{M^3}\right)\oneb_{\lambda'_\mathrm{max}\le M/2}\right]. 
\end{align}
Using
\begin{align}\label{eqn:tailm}
\E\left[\sum_{j=1}^\q(\lambda_j')^m \mathbf{1}_{\max\lambda_j'>M/2}\right]=O(e^{-cM(1-o(1))}),\quad m=1,2,3,
\end{align}
which follows from H\"older's inequality and Lemmas~\ref{lem:gsv} and \ref{lem:g-maxsv}, we can also apply Lemma~\ref{lem:gsv} to \eqref{eqn:log-lower} to obtain
\begin{align}
\E\sum_{j=1}^\q\log\Big(1-\frac{\lambda_j'}{M}\Big)\oneb_S(\Gqq)&\ge -\frac{\q^2+\q}{M} - \frac{2\q^3+O(\q^2)}{2M^2}- O\left(\frac{\q^4}{M^3}\right)- O(e^{-cM(1-o(1))}).
\end{align}
Inserting this back in \eqref{eqn:logratio2} with the expectation, and applying \eqref{eqn:tailm} to the trace term as well, gives the upper bound
\begin{align*}
\E\left[\oneb_{S}(\Gqq)\log\frac{g(\Gqq)}{\tilde f(\Gqq)}\right]&
\le
\begin{multlined}[t]
\frac{\q^3}{2M}+ O(e^{-cM})
-\frac{M-2\q-1}{2}\left(-\frac{\q^2+\q}{M}-\frac{2\q^3+O(\q^2)}{2M^2}\right)\\
-\frac{\q^2+\q}{2} +O\left(\frac{\q^2}{M}\right)+O\left(\frac{\q^4}{M^2}\right)
\end{multlined}\\
&= O\left(\frac{\q^2}{M}\right).\numberthis
\end{align*}
Combining this with \eqref{eqn:maxsv} and \eqref{eqn:amaxsv} in \eqref{eqn:zeta2} implies
\begin{align}
\zeta \ge 1-O(\q^2/M).
\end{align}
Thus $|1-\zeta|=O(\q^2/M)$, which also implies $|\log\zeta|=O(\q^2/M)$.
\end{proof}

\begin{proof}[Proof of Theorem~\ref{thm:hiding}]
Combining Propositions~\ref{prop:exp} and \ref{prop:zeta}, and noting it is sufficient to consider $N\le c_0\sqrt{M}$ since TVD is bounded above by 1, we obtain from Pinsker's inequality \eqref{eqn:pinsker}, and \eqref{eqn:dkl-split} that
\begin{align}
d_\TV(\sqrt{M}\Uqq,\Gqq)&\le\sqrt{\frac{1}{2}\dkl(\sqrt{M}\Uqq\;||\;\Gqq)}=O(\q/\sqrt{M}),
\end{align}
as desired. 
\end{proof}

\section{Proof of lemmas}\label{sec:lemmas}
\subsection{Singular value bounds}\label{subsec:g-sv}
In this section we prove the lemmas used in the proof of Proposition~\ref{prop:zeta}.
\begin{lem}\label{lem:gsv}
Let $\Gqq\sim\gsym$, and let $\lambda_1',\ldots,\lambda_\q'$ be the eigenvalues of $\Gqq^\dagger\Gqq$. Then
\begin{align}
&\E \sum_{j=1}^\q \lambda_j' =\q^2+\q,\quad
\E \sum_{j=1}^\q (\lambda_j')^2=2\q^3+O(\q^2),\\
&\E \sum_{j=1}^\q (\lambda_j')^3=O(\q^4),\quad 
\E \sum_{j=1}^\q (\lambda_j')^4=O(\poly(\q)).\label{eqn:g-34}
\end{align}
\end{lem}
\begin{proof}
We calculate these directly. For example,
\begin{align}\label{eqn:gsv-trace}
\E\sum_{j=1}^\q \lambda_j' = \E\Tr(\Gqq^\dagger\Gqq)&=\sum_{i_1,i_2=1}^\q\E|\Gqq_{i_1i_2}|^2
=2\q+(\q^2-\q)=\q^2+\q.
\end{align}
Using Isserlis's or Wick's formula for complex Gaussians (see e.g. \cite[Ch.2]{Speicher}), and recalling that $\Gqq$ has a different variance on the diagonal,
\begin{align*}
\E\sum_{j=1}^\q(\lambda_j')^2=\E\Tr(\Gqq^\dagger\Gqq\Gqq^\dagger\Gqq)&=\sum_{i,j,k,\ell=1}^\q\E[\Gqq_{ij}\bar\Gqq_{kj}\Gqq_{k\ell}\bar\Gqq_{i\ell}]\\
&=2\sum_{i,j,k,\ell=1}^\q\left(4\delta_{i=j=k=\ell}+2\delta_{j=k=\ell,\ne i}+2\delta_{i=j=\ell,\ne k}+\delta_{j=\ell,i\ne j,k\ne j}\right)\\
&=2\left[4\q+4\q(\q-1)+\q(\q-1)^2\right]=2\q^3+4\q^2+2\q.\numberthis
\end{align*}
Similarly, applying Isserlis's or Wick's formula yields \eqref{eqn:g-34}. Note we only need the leading order up to constants for the third moment, and for the fourth moment the $O(\q^8)$ from H\"older's inequality\footnote{$\E\sum_{j=1}^N(\lambda_j')^4=\E\Tr[(\Gqq^\dagger\Gqq)^4]\le\sum_{i_1,\ldots,i_8=1}^N \E[|\Gqq_{i_1i_2}|^8]^{1/8}\cdots\E[|\Gqq_{i_8i_1}|^8]^{1/8}\le CN^8$.
} is sufficient.
\end{proof}

The next lemma is a more precise (though far from optimal) estimate for the corresponding quantity $g_\mathrm{tail}$ of \cite[Lemma 5.3]{aa}. For the case of \cite{aa}, note the Gaussian matrix there is not symmetric, and there $g_\mathrm{tail}$ concerns the largest eigenvalue of a complex Wishart distribution. 
\begin{lem}\label{lem:g-maxsv}
Let $\lambda_1',\ldots,\lambda_\q'$ be the eigenvalues of $\Gqq^\dagger\Gqq$, where $\Gqq\sim\gsym$. 
Then for $t>2\q^2$,
\begin{align}
\P[\max(\lambda_1',\ldots,\lambda_\q')>t]&\le e^{-c\left(\sqrt{t}-\sqrt{2}\q\right)^2}.
\end{align}
\end{lem}
\begin{proof}
Recall that the operator norm $\|\Gqq\|$ is equal to the largest singular value, which is $\sqrt{\max(\lambda_1',\ldots,\lambda_\q')}$.
Using subgaussian concentration of the operator norm $\|\Gqq\|$ since the norm is $\sqrt{2}$-Lipschitz for symmetric matrices (compare to Frobenius norm), we obtain
\begin{align}
\P[\|\Gqq\|>\E\|\Gqq\|+s]&\le e^{-cs^2}, \quad s>0;
\end{align}
see for example \cite[Theorem 2.26]{Wainwright}, noting that we can view a function of $n$ i.i.d. standard complex Gausian random variables as a function of $2n$ i.i.d. real Gaussian random variables. Since
\begin{align}
\E\|\Gqq\|&\le (\E\|\Gqq\|^2)^{1/2}=(\E\lambda'_\mathrm{max})^{1/2}\le \left(\E\sum_{j=1}^\q\lambda_j'\right)^{1/2}=(\q^2+\q)^{1/2}\le \sqrt{2}\q,
\end{align}
where we used Lemma~\ref{lem:gsv} for the bound on the expected trace, then letting ${t}=(\E\|\Gqq\|+s)^2$,
\begin{align*}
\P\left[\max(\lambda_1',\ldots,\lambda_\q')>t\right]=\P\left[\|\Gqq\|>\sqrt{t}\right]&\le e^{-c\left(\sqrt{t}-\sqrt{2}\q\right)^2}.
\end{align*}
\end{proof}

\begin{lem}\label{lem:u-maxsv}
Let $\lambda_1,\ldots\lambda_\q$ be the eigenvalues of $\Uqq^\dagger\Uqq$. Then for $\q=O(\sqrt{M})$,
\begin{align}
\P[\max(\lambda_1,\ldots,\lambda_\q)>1/2] &\le O(\q^2/M).
\end{align}
\end{lem}
\begin{proof}
By Markov's inequality and Lemma~\ref{lem:submatrix-sv}, we have
\begin{align*}
\P\left[\max(\lambda_1,\ldots,\lambda_\q)>\frac{1}{2}\right]&\le\P\left[\sum_{j=1}^\q\lambda_j>\frac{1}{2}\right]\le  2\E\left[\sum_{j=1}^\q\lambda_j\right]=O\left(\frac{\q^2}{M}\right).
\end{align*}
\end{proof}

\subsection{Proof of Lemma~\ref{lem:submatrix-sv}}\label{sec:coe-moments}

In this section we prove Lemma~\ref{lem:submatrix-sv} on the moments of the singular values of $\Uqq$.
Recall $\Uqq$ is the upper left $N\times N$ submatrix of an $M\times M$ COE matrix $U\deq VV^T$, where $V$ is a Haar random unitary matrix.
We are interested in $\E\sum_{j=1}^N\lambda_j^m$ for $m=1,2,3$.

For $m=1$, we can apply Weingarten calculus \cite{Collins2003,CollinsSniady2006} or the fourth moment calculation given in e.g. \cite[Lemma 4.7]{ChatterjeeMeckes2008}, giving
\begin{align*}
\E\sum_{j=1}^\q\lambda_j=\E\Tr(\Uqq^\dagger\Uqq)=\sum_{i_1,i_2=1}^\q \E|\Uqq_{i_1i_2}|^2
&=\sum_{i_1,i_2=1}^\q\sum_{j_1,j_2=1}^M\E[V_{i_1j_1}V^T_{j_1i_2}\bar V_{i_1j_2}\bar V^T_{j_2i_2}]\\
&=M(\q^2+\q)\left[\frac{1}{(M-1)(M+1)}-\frac{1}{(M-1)M(M+1)}\right]\\
&=\frac{\q^2+\q}{M+1},\numberthis 
\end{align*}
which is $\frac{\q^2+\q}{M}+O\left(\frac{\q^2}{M^2}\right)$. 

For $m=2,3$, we only need the leading order term, which will simplify calculations.
In order to compute expectation values, we use the Weingarten calculus \cite{Collins2003,CollinsSniady2006} with only the leading order asymptotics for the Weingarten function:
\begin{thm}[{\cite[Corollaries 2.4, 2.7]{CollinsSniady2006}}]\label{thm:weingarten}
Let $\mathcal U(d)$ be the space of $d\times d$ unitary matrices equipped with Haar measure, and let $S_n$ be the permutation group on $n$ elements. Let $n$ be a positive integer and $i=(i_1,\ldots,i_n)$, $i'=(i_1',\ldots,i_n')$, $j=(j_1,\ldots,j_n)$, $j'=(j_1',\ldots,j_n')$ be $n$-tuples of positive integers. Then
\begin{align}\label{eqn:weingarten}
\int_{\mathcal{U}(d)}U_{i_1j_1}\cdots U_{i_nj_n}\overline{U_{i_1'j_1'}}\cdots\overline{U_{i_n'j_n'}}\,dU
&= \sum_{\sigma,\tau\in S_n}\delta_{i_1i'_{\sigma(1)}}\cdots\delta_{i_ni'_{\sigma(n)}}\delta_{j_1j'_{\tau(1)}}\cdots\delta_{j_nj'_{\tau(n)}} \Wg(\tau\sigma^{-1}),
\end{align}
where $\Wg$ denotes the Weingarten function defined in \cite[Eq.~(9)]{CollinsSniady2006}.
The Weingarten function satisfies the asymptotics, for fixed $n$ as $d\to\infty$,
\begin{align}\label{eqn:weingarten-asymptotics}
\Wg(\operatorname{Id})&=\frac{1}{d^n}\left(1+O\left(\frac{1}{d^2}\right)\right),\quad
\Wg(\sigma)=O\left(\frac{1}{d^{n+|\sigma|}}\right),
\end{align}
where $|\sigma|$ denotes the minimum number of factors needed to write $\sigma$ as a product of transpositions, or equivalently, $|\sigma|=n-(\#\text{ of disjoint cycles})$.
\end{thm}
To apply the above theorem, write
\begin{align*}
\E\sum_{j=1}^\q\lambda_j^2=\E\Tr(\Uqq^\dagger\Uqq\Uqq^\dagger\Uqq)&=\sum_{i_1,\ldots,i_4=1}^\q\E[\bar\Uqq_{i_2i_1}\Uqq_{i_2i_3}\bar\Uqq_{i_4i_3}\Uqq_{i_4i_1}]\\
&=\sum_{i_1,\ldots,i_4=1}^\q\sum_{j_1,j_2,j_1',j_3'=1}^M\E[V_{i_1j_1}V_{i_2j_2}V_{i_3j_2}V_{i_4j_1}\bar V_{i_1j_1'}\bar V_{i_2j_1'}\bar V_{i_3j_3'}\bar V_{i_4j_3'}],
\end{align*}
where we use indices $j_1',j_3'$ to match the notation in Theorem~\ref{thm:weingarten}. To match more exactly, we could instead write the sum over all $j_1,j_2,j_3,j_4,j_1',j_2',j_3',j_4'$ by adding the Kronecker deltas $\delta_{j_1j_4}\delta_{j_2j_3}\delta_{j_1'j_2'}\delta_{j_3'j_4'}$. Applying Theorem~\ref{thm:weingarten}, the above expression becomes
\begin{align*}
\E\sum_{j=1}^\q\lambda_i^2&=\sum_{i_1,\ldots,i_4=1}^\q\sum_{j_1,j_2,j_1',j_3'=1}^M\sum_{\sigma,\tau\in S_4}\Wg(\tau\sigma^{-1})\delta_{i_1i_{\sigma(1)}}\delta_{i_2i_{\sigma(2)}}\delta_{i_3i_{\sigma(3)}}\delta_{i_4i_{\sigma(4)}}\delta_{j_1j_{\tau(1)}'}\delta_{j_2j'_{\tau(2)}}\delta_{j_2j'_{\tau(3)}}\delta_{j_1j'_{\tau(4)}},
\end{align*}
where for the indices $j'_{\tau(k)}$, we identify the skipped indices $j'_2,j'_4$ via $j'_2\equiv j'_1$ and $j'_4\equiv j'_3$, due to the repeated primed indices in the terms $\bar V_{i_kj_\ell'}$. 

Evaluating the sums over $i_1,\ldots,i_4$ yields
\begin{align}\label{eqn:wgdelta4}
\E\sum_{j=1}^\q\lambda_i^2&=\sum_{\sigma,\tau\in S_4}\q^{\#\text{cycles}(\sigma)}\Wg(\tau\sigma^{-1})\sum_{j_1,j_2,j_1',j_3'=1}^M\delta_{j_1j_{\tau(1)}'}\delta_{j_2j'_{\tau(2)}}\delta_{j_2j'_{\tau(3)}}\delta_{j_1j'_{\tau(4)}}.
\end{align}
We can evaluate
\begin{align}\label{eqn:delta4}
\sum_{j_1,j_2,j_1',j_3'=1}^M\delta_{j_1j_{\tau(1)}'}\delta_{j_2j'_{\tau(2)}}\delta_{j_2j'_{\tau(3)}}\delta_{j_1j'_{\tau(4)}}&=\begin{cases}
M^2,&\text{$\tau$ maps $\{1,4\}$ to $\{1,2\}$ or $\{3,4\}$ (write, $\tau\in T$)}\\
M,&\text{otherwise}
\end{cases},
\end{align}
according to Figure~\ref{fig:j4}: if $\tau$ maps $\{1,4\}$ to $\{1,2\}$ or $\{3,4\}$, which we will denote as $\tau\in T$, then the graph in Figure~\ref{fig:j4} has two cycles, but otherwise has only one cycle. Each cycle represents indices that all must take the same value, so we obtain a factor of $M$ for each distinct cycle after summing over the indices.
\begin{figure}[htb]
\centering
\includegraphics{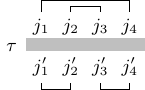}
\caption{Indices $j_k$ and $j'_\ell$ that must take identical values are shown connected with a bracket. In particular, $j_1= j_4$ and $j_2= j_3$, and $j_1'= j_2'$ and $j_3'= j_4'$. The permutation $\tau$ combined with the Kronecker $\delta$ functions connects the two rows of $(j_k)$'s and $(j_k')$'s. This forms a graph with the indices $(j_k),(j_k')$ as the vertices, and where indices in the same cycle must take the same value. 
}\label{fig:j4}
\end{figure}

Now we start to consider leading order terms. Using \eqref{eqn:weingarten-asymptotics}, the Weingarten function $\Wg(\tau\sigma^{-1})$ obtains the largest order $1/M^4$ only for $\tau\sigma^{-1}=\operatorname{Id}$, and is $O(1/M^5)$ otherwise. Using \eqref{eqn:delta4}, the largest order terms with $\tau\sigma^{-1}=\operatorname{Id}$, or $\sigma=\tau$, occur when 
\begin{align}\label{eqn:tau4}
\tau\in\{(2\,4),(1\,3)\}\subset T,
\end{align}
which each lead to a term of order $\q^3/M^2$. To see that other choices of $\sigma=\tau$ will produce a lower order term, first consider $\tau\in T$. If $\tau$ has a cycle shape that is not $[2,1,1]$ or $[1,1,1,1]$, then the power $\q^{\#\text{cycles}(\sigma)}$ will be a lower order, $\q^2$ or smaller. The only cycles in $T$ with shape $[2,1,1]$ are the two in \eqref{eqn:tau4}, and the only cycle with shape $[1,1,1,1]$ is the identity which is not in $T$. Next, if $\tau\not\in T$, then \eqref{eqn:delta4} implies the highest order we could obtain is $\q^4/M^3$, which is $O(\q^2/M^2)$ for $\q=O(\sqrt{M})$. Finally, if $\sigma\ne\tau$, then $\Wg(\tau\sigma^{-1})=O(1/M^5)$, giving largest possible order $\q^4/M^3=O(\q^2/M^2)$. This proves \eqref{eqn:lambda2}.

In the case $m=3$, we similarly have
\begin{align*}
\E\sum_{j=1}^\q\lambda_j^3&=\E\Tr(\Uqq^\dagger\Uqq\Uqq^\dagger\Uqq\Uqq^\dagger\Uqq)\\
&=\sum_{i_1,\ldots,i_6=1}^\q\sum_{\substack{j_1,j_2,j_4,\\j_1',j_3',j_5'=1}}^M\E[V_{i_1j_1}V_{i_2j_2}V_{i_3j_2}V_{i_4j_4}V_{i_5j_4}V_{i_6j_1}\bar V_{i_1j_1'}\bar V_{i_2j_1'}\bar V_{i_3j_3'}\bar V_{i_4j_3'}\bar V_{i_5j_5'}\bar V_{i_6j_5'}]\\
&=\sum_{\sigma,\tau\in S_6}\q^{\#\text{cycles}(\sigma)}\Wg(\tau\sigma^{-1})\sum_{\substack{j_1,j_2,j_4,\\j_1',j_3',j_5'=1}}^M(\delta_{j_1j_{\tau(1)}'}\delta_{j_2j'_{\tau(2)}}\delta_{j_2j'_{\tau(3)}}\delta_{j_4j'_{\tau(4)}}\delta_{j_4j'_{\tau(5)}}\delta_{j_1j'_{\tau(6)}})\delta_{j_1'j_2'}\delta_{j_3'j_4'}\delta_{j_5'j_6'}\\
&=\sum_{\sigma,\tau\in S_6}\q^{\#\text{cycles}(\sigma)}\Wg(\tau\sigma^{-1})\begin{cases}
M^3,&\tau(\{1,6\}),\tau(\{2,3\}),\tau(\{4,5\})\in \text{ some }A_i\\
M^2\text{ or }M,&\text{otherwise}
\end{cases},\numberthis\label{eqn:wgdelta6}
\end{align*}
where $A_1=\{1,2\}$, $A_2=\{3,4\}$, and $A_3=\{5,6\}$,
using Figure~\ref{fig:j6} to obtain the last line.
\begin{figure}[htb]
\centering
\includegraphics{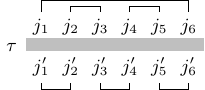}
\caption{Analogue of Fig.~\ref{fig:j4} for $m=3$. Indices of $j$ and $j'$ that must take identical values are connected with a bracket. The permutation $\tau$ connected the two rows of $(j_k)$'s and $(j_k')$'s, so that indices in the same cycle of the resulting graph must take the same value. There is only one independent index for each cycle, leading to one factor of $M$ for each cycle after the summation.}\label{fig:j6}
\end{figure}
Reasoning as in the $m=2$ case, for the leading order we first consider $\sigma=\tau$ with $\tau$ contributing a factor $M^3$ in \eqref{eqn:wgdelta6}. There are no such $\tau$ with cycle shape $[1,1,1,1,1,1]$ or $[2,1,1,1,1]$, and the shapes with the next most cycles are $[2,2,1,1]$ and $[3,1,1,1]$, which will give a term $\q^4/M^3$.
For example,
\begin{align}
\tau\in\{(1\,3)(4\,6),(2\,6)(3\,5),(1\,5)(2\,4),(1\,5\,3),(2\,4\,6)\},
\end{align}
produce this term.
Any other shapes of $\tau$ or $\sigma$ will give at most order $\q^6/M^4$ or $\q^3/M^3$, which are also $O(\q^4/M^3)$ for $\q=O(\sqrt{M})$. \qed

\section{Proof of Corollary~\ref{cor:ggt} and Theorem~\ref{thm:instance}}\label{sec:instance}

\subsection{Proof of Corollary~\ref{cor:ggt}}
Let $Z$ be a $K\times K$ matrix of i.i.d. $\CN(0,1/\sqrt{K})$ random variables. 
A Haar-distributed random unitary matrix can be formed \cite{mezzadri} from $Z$ by performing Gram--Schmidt orthonormalization to the rows of $Z$, creating a $K\times K$ unitary matrix $V$. The matrix $W:=VV^T$ is then distributed as a $K\times K$ COE matrix. Theorem 1 of \cite{Jiang2009} shows that for this construction,
\begin{align}\label{eqn:coeyyt}
\P\left[\max_{1\le i,j\le N}\left|\sqrt{K}W_{ij}-(ZZ^T)_{ij}\right|>t\right] &\le CK^2e^{-ct\sqrt{K}/N},
\end{align}
for any $N^2\le K/12$ and $0<t\le4\sqrt{K}/N$. This means that in terms of the maximum entrywise distance, the $N\times N$ top left submatrix of the COE matrix $W$ after rescaling looks very close to the $N\times N$ top left submatrix of the corresponding symmetric product $ZZ^T$ for the Gaussian matrix $Z$.

Let $G$ be the top $N\times K$ submatrix of $Z$, so that the top left $N\times N$ submatrix of $ZZ^T$ is $GG^T\dsim\ggtnk$, and let $W_{NN}$ denote the top left $N\times N$ submatrix of $W$. Due to \eqref{eqn:coeyyt}, it will suffice to show that there is $\Gqq\dsim\gsym$ with entries close to $\sqrt{K}W_{NN}$. 
We apply Theorem~\ref{thm:hiding} with $M=K$, which implies that for $N=o(\sqrt{K})$,
\begin{align}\label{eqn:kv}
d_\TV(\sqrt{K}W_{NN},\Gqq)&\le O(N/\sqrt{K}).
\end{align}
The total variation distance between two distributions $\mu$ and $\nu$ on a Polish space is related to \emph{couplings} of $\mu$ and $\nu$, which are defined as jointly distributed random variables $(X,Y)$ such that $X\dsim\mu$ and $Y\dsim\nu$. In particular, it is known that there always exists an optimal coupling $(X,Y)$ such that
\begin{align}
\P[X\ne Y]=d_\TV(\mu,\nu);
\end{align}
see for example \cite[Ch.1 Theorem 5.2]{Lindvall1992}.
Applying this with \eqref{eqn:kv}, there is a joint distribution $\rho$ of random variables $(X,\Gqq)$ such that $X\deq\sqrt{K}W_{NN}$, $\Gqq\dsim\gsym$, and
\begin{align}
\P[X\ne \Gqq]\le O(N/\sqrt{K}).
\end{align}

Now construct $GG^T$, $\sqrt{K}W_{NN}$, and $\Gqq$ as follows. Given $G$ an $N\times K$ matrix of i.i.d. $\CN(0,1/\sqrt{K})$ random variables, construct $\sqrt{K}W_{NN}$ via the Gram--Schmidt orthonormalization applied to the rows of $G$. Given $X=\sqrt{K}W_{NN}$, sample $\Gqq$ according to the conditional distribution $\rho_{\Gqq|X}$ (see e.g. \cite[\S IV.2]{Cinlar-book}). The pair $(\sqrt{K}W_{NN},\Gqq)$ then has joint distribution $\rho$, so $\Gqq\dsim\gsym$, and for this construction of $GG^T$, $\sqrt{K}W_{NN}$, and $\Gqq$, we have
\begin{align*}
\P\left[\max_{1\le i,j\le N}\left|\Gqq_{ij}-(GG^T)_{ij}\right|>\varepsilon\right] &\le
\begin{multlined}[t]
\P\left[\max_{1\le i,j\le N}\left|\Gqq_{ij}-(\sqrt{K}W_{NN})_{ij}\right|>\varepsilon/2\right]+\\+\P\left[\max_{1\le i,j\le N}\left|(\sqrt{K}W_{NN})_{ij}-(GG^T)_{ij}\right|>\varepsilon/2\right]
\end{multlined}\\
&\le O(N/\sqrt{K})+ CK^2e^{-c\varepsilon\sqrt{K}/N},\numberthis
\end{align*}
which is $o(1)$ as $K\to\infty$ for $N=o(\sqrt{K}/\log K)$.
\qed

\subsection{Proof of Theorem~\ref{thm:instance}}

The proof is the same as that of \cite[Theorem 5.2]{aa}, using Proposition~\ref{prop:zeta} to ensure $\zeta$ is close to one. For completeness, we write the proof here for this case. Letting $f$ be the density of $\sqrt{M}\Uqq$ and $g$ the density of $\Gqq$, we want to bound $\|f/g\|_\infty$. Recalling that $\tilde f=\zeta f$ for $\zeta$ given in \eqref{eqn:zeta}, we have
\begin{align*}
\left\|\frac{f}{g}\right\|_\infty&=\frac{1}{\zeta}\sup_{X\in\C^{\q(\q+1)/2}}\frac{\tilde f(X)}{g(X)}\\
&=\begin{multlined}[t]
\frac{1}{\zeta}\exp\left(-\frac{\q^3}{2M}\right)
\sup_{\lambda_1,\ldots,\lambda_\q\ge0}\exp\left(\frac{1}{2}\sum_{j=1}^\q\left[(M-2\q-1)\log\left(1-\frac{\lambda_j}{M}\right)+\lambda_j\right]\right)\oneb_{\lambda_\mathrm{max}\le M}.\numberthis
\end{multlined}
\end{align*}
Since the function $\lambda\mapsto (M-2\q-1)\log\left(1-\frac{\lambda}{M}\right)+\lambda$ achieves its maximum at $\lambda=2\q+1$, which is less than $M$, we obtain 
\begin{align*}
\left\|\frac{f}{g}\right\|_\infty&=
\frac{1}{\zeta}\exp\left(-\frac{\q^3}{2M}\right)
\left(1-\frac{2\q+1}{M}\right)^{\q(M-2\q-1)/2}e^{\q^2+\frac12\q}\\
&=\frac{e^{O(\q^3/M)}}{1+O(\q^2/M)} = 1+O(\q^3/M).\numberthis
\end{align*}
\qed

\section{Discussion}\label{sec:conclusion}
In this paper, we prove Theorem~\ref{thm:hiding-conj}/\ref{thm:hiding}, which proves the hiding conjecture for Gaussian boson sampling in the setting where all input modes are squeezed ($K=M$). 
The relevant Gaussian matrix in the theorem is $\Gqq\dsim\gsym$ (symmetric complex Gaussian with independent entries modulo symmetry) rather than the symmetric product of Gaussians $GG^T\dsim\ggtnk$, which as we explain in Remark~\ref{rmk:11} appears a more favorable ensemble to work with due to its simpler structure and independent entries (modulo symmetry).
Secondarily, we prove Theorems~\ref{thm:sparse-hiding0} and \ref{thm:sparse-hiding} on hiding for $NK=o(M)$, which covers the sparse squeezer regime. This also confirms the extension for hiding in Fock boson sampling to the maximal submatrix size $o(\sqrt{M})\times o(\sqrt{M})$ as expected in \cite{aa}.
We note this paper leaves completely open the hiding conjecture for $K=cM$ with $K<M$. This has since been resolved in follow-up work \cite{hiding2}.

Finally, we note that as is typical in many theoretical arguments for Fock or Gaussian boson sampling, including in the original boson sampling proposal of \cite{aa}, we work in the dilute or collision-free regime, where the number of photons is much smaller than the number of modes, and one makes use of a hiding property or conjecture.
Current large-scale experimental realizations of GBS \cite{zhong2020quantum,Zhong2021,madsen2022,deng2023gaussian,liu2026gaussian} generally work with larger photon numbers which can be linear in the number of modes, in which case the output distribution no longer falls in the dilute or collision-free regime. A general area of interest in both Fock and Gaussian boson sampling is thus to investigate properties and statistics in the non-dilute regime, as studied for example in \cite{grier2022complexity,bouland2023complexity}.

\appendix 

\section{Sparse case and hiding in Fock boson sampling}\label{app:sparse}

In this section, we give the proof of Theorem~\ref{thm:sparse-hiding0}, which extends the sparse case of hiding in GBS from $K=o(M^{1/5})$ to its maximal extent $K=o(M)$, in particular allowing for $K=O(M^{1-\varepsilon})$, any $\varepsilon>0$.
In regards to the hiding property for Fock boson sampling \cite{aa}, we remark that it appears not that well-known in the physics literature that \cite{Jiang2009jacobi} obtains the maximal $o(\sqrt{M})\times o(\sqrt{M})$ size for a square submatrix of a Haar unitary to be close to Gaussian in TV distance (though without a quantitative rate of convergence), or that the later method in \cite{JiangMa2019} can be used to give a quantitative rate of convergence. 
In particular, as we describe below, the method of \cite{JiangMa2019} leads to the maximal $o(\sqrt{M})\times o(\sqrt{M})$ submatrix size, with required error bound, conjectured in \cite{aa} for hiding in Fock boson sampling.
For the above reasons, we include most of the details here, noting that only the few lemmas written below are needed for the proof.
The proof simply follows along the steps in \cite{JiangMa2019}, which proves the maximal rectangular submatrix sizes for Haar-distributed real orthogonal matrices. We consolidate Theorems~\ref{thm:sparse-hiding0} and \ref{thm:fock-hiding} into the following theorem.
\begin{thm}\label{thm:sparse-hiding}
Let $U_{NK}$ be the upper $N\times K$ submatrix of an $M\times M$ Haar unitary matrix, and let $G_{NK}$ be an $N\times K$ matrix of i.i.d. standard complex Gaussians. Then for $NK=o(M)$,
\begin{align}\label{eqn:tv-pq}
d_\TV(\sqrt{M}U_{NK},G_{NK})&=O\left(\sqrt{\frac{NK}{M}}\right).
\end{align}
As a consequence, the following hiding property for Gaussian boson sampling holds for any $NK=o(M)$ with $N\le K$:
\begin{align}\label{eqn:tv-sparse}
d_\TV(MU_{NK}U_{NK}^T, G_{NK}G_{NK}^T) &= O\left(\sqrt{\frac{NK}{M}}\right).
\end{align}
\end{thm}

In what follows, we will switch to considering a $p\times q$ submatrix with $p\ge q$, in order to follow along better with \cite{JiangMa2019}.
Since a Haar unitary and its transpose have the same distribution, one can simply consider taking a transpose and setting $p=K$ and $q=N$.

The form of the density (without normalization constant) for the submatrix of a Haar unitary matrix has been derived by various methods, such as in \cite{Eaton1989,ZS2000,collins2003thesis,Forrester-book}. The normalizing constant can be calculated using various methods, such as the matrix integrals of \cite[Theorem 2.2.1]{Hua1963} or via the Selberg integral \cite{Selberg1944}.
\begin{prop}[see {\cite[Prop. 2.1]{Jiang2009jacobi}}]\label{prop:unitary-density}
Let $\bpq$ be the upper left $p\times q$ submatrix of an $M\times M$ Haar random unitary matrix. Then for $p\ge q$ and $p+q\le M$, the density of $\bpq$ is
\begin{align}
f(Z)&=\pi^{-pq}\prod_{j=1}^{q}\frac{(M-j)!}{(M-j-p)!}\det(I_q-Z^\dagger Z)^{M-q-p}\iz(\lambda_\mathrm{max}(Z^\dagger Z)<1),
\end{align}
where $\iz(\cdot)$ denotes an indicator function and $\lambda_\mathrm{max}(Z^\dagger Z)$ the maximum eigenvalue of $Z^\dagger Z$.
\end{prop}
\begin{rmk}
For the square case $p=q$, after rescaling to consider the density of $\sqrt{M}\bpq$ as written in \eqref{eqn:sparse-density}, we can compare the resulting normalization constant $c_{M,q}=\pi^{-q^2}M^{-q^2}\prod_{j=1}^q\frac{(M-j)!}{(M-j-q)!}$ to the approximation used in \cite{aa}. There they define $\zeta:=\pi^{-q^2}/c_{M,q}$, and prove that $|\zeta-1|$ is small, without knowing the exact form of $c_{M,q}$. Using the expansion \eqref{eqn:kn} below, we see that
\begin{align*}
\zeta&=\exp\left(O(q^3/M)\right),
\end{align*}
which is close to 1 if $q=o(M^{1/3})$, which is more restrictive than the size in Theorem~\ref{thm:sparse-hiding}.
One would need to use a better approximation or guess for $\zeta$, based on matching cancellations in the proof method, to allow for larger size submatrices without knowing $c_{M,q}$.
\end{rmk}

For the proof of Theorem~\ref{thm:sparse-hiding}, we only need two additional lemmas. First we need the asymptotics of the normalizing constant. Following the same steps of Lemma 2.7 of \cite{JiangMa2019} and keeping track of the size of the $o(1)$ error term gives (for $p=p_M\to\infty$; the case $p=O(1)$ follows by Taylor expansion):
\begin{lem}\label{lem:kn}
Let $1\le q\le p<M$. For $\limsup \frac{p}{M}<1$ and $pq=O(M)$,
\begin{align}\label{eqn:kn}
\log\left[\frac{1}{M^{pq}}\prod_{j=1}^{q}\frac{(M-j)!}{(M-j-p)!}\right]
=-pq-q\left(M-p-\frac{q}{2}\right)\log\left(1-\frac{p}{M}\right) +O\left(\frac{p^2q^2}{M^2}\right),
\end{align}
as $M\to\infty$.
\end{lem}
We also need singular value moment estimates.
\begin{lem}\label{lem:un-sv}
Let $\bpq$ be the upper left $p\times q$ submatrix, $p\ge q$, of an $M\times M$ Haar random unitary matrix. Letting $\lambda_1,\ldots,\lambda_q$ be the nonzero eigenvalues of $\bpq^\dagger\bpq$, then for $pq=o(M)$,
\begin{align}
\E\sum_{j=1}^q\lambda_j&=\frac{pq}{M},\quad
\E\sum_{j=1}^q\lambda_j^2=\frac{pq(p+q)}{M^2}+O\left(\frac{pq}{M^2}\right),\quad
\E\sum_{j=1}^q\lambda_j^3=\frac{p^3q}{M^3}+O\left(\frac{p^2q^2}{M^3}\right).
\end{align}
\end{lem}
\begin{proof}
This follows using Weingarten calculus (or \cite[Lemma 4.7]{ChatterjeeMeckes2008} for $\E\sum_j\lambda_j^2$), with simpler calculations than for the COE case of Lemma~\ref{lem:submatrix-sv}.
\end{proof}

\begin{proof}[Proof of Theorem~\ref{thm:sparse-hiding}]
The density for $G$ is $g(Z)=\pi^{-pq}e^{-\Tr(Z^\dagger Z)}$. By Proposition~\ref{prop:unitary-density}, the density of $\sqrt{M}\bpq$ is
\begin{align}\label{eqn:sparse-density}
f(Z)&=M^{-pq}\pi^{-pq}\prod_{j=1}^{q}\frac{(M-j)!}{(M-j-p)!}\det\left(I_q-\frac{Z^\dagger Z}{M}\right)^{M-q-p}\iz(\lambda_\mathrm{max}(Z^\dagger Z)< M),
\end{align} 
recalling that there are $pq$ complex variables of integration, or $2pq$ real variables, for the rescaling.
Following exactly the proof of \cite[Theorem 1(i)]{JiangMa2019} using the KL divergence, we apply Lemma~\ref{lem:kn} to obtain for $\lambda_\mathrm{max}(Z^\dagger Z)<M$,
\begin{multline}
\log\frac{f(Z)}{g(Z)}= 
-pq-q\left(M-p-\frac{q}{2}\right)\log\left(1-\frac{p}{M}\right) +O\left(\frac{p^2q^2}{M^2}\right)+\\
+ \sum_{j=1}^q(M-q-p)\log\left(1-\frac{\lambda_j(Z^\dagger Z)}{M}\right)+\sum_{j=1}^q\lambda_j(Z^\dagger Z).
\end{multline}
Taking the expectation with respect to $\sqrt{M}\bpq$ and applying the first equality of Lemma~\ref{lem:un-sv} (recall we rescaled $\lambda_j\mapsto M\lambda_j$) yields
\begin{align*}
\E\left[\log\frac{f(\sqrt{M}\bpq)}{g(\sqrt{M}\bpq)}\right]&= \begin{multlined}[t]
-\frac{q^2}{2}\log\left(1-\frac{p}{M}\right)
+ (M-q-p)\E\sum_{j=1}^q\log\left(\frac{1-\frac{\lambda_j}{M}}{1-\frac{p}{M}}\right)+O\left(\frac{p^2q^2}{M^2}\right)
\end{multlined}.\numberthis
\end{align*}
The logarithm term inside the sum can be written as $\log\left(1+\frac{p-\lambda_j}{M-p}\right)$.
Then using the inequality $\log(1+x)\le x-x^2/2+x^3/3$ for any $x>-1$, we see
\begin{align}
\begin{aligned}
\E\left[\log\frac{f(\sqrt{M}\bpq)}{g(\sqrt{M}\bpq)}\right]&\le
\begin{multlined}[t]
\frac{q^2}{2}\log\left(1+\frac{p}{M-p}\right)
\\ +(M-q-p)\E\sum_{j=1}^q\left[\frac{p-\lambda_j}{M-p}-\frac{(p-\lambda_j)^2}{2(M-p)^2}+\frac{(p-\lambda_j)^3}{3(M-p)^3}\right]+O\left(\frac{p^2q^2}{M^2}\right).
\end{multlined}
\end{aligned}
\end{align}
Finally, applying Lemma~\ref{lem:un-sv} we obtain
\begin{align}
\begin{aligned}
\E\left[\log\frac{f(\sqrt{M}\bpq)}{g(\sqrt{M}\bpq)}\right]&\le
\begin{multlined}[t]
\frac{q^2}{2}\frac{p}{M-p}+(M-q-p)\bigg[-\frac{p^2q-2p^2q+pq(p+q)+O(pq)}{2(M-p)^2}+\\
\frac{p^3q-3p^2(pq)+3p^2q(p+q)-p^3q+O(p^2q^2)}{3(M-p)^3}\bigg]+O\left(\frac{p^2q^2}{M^2}\right)
\end{multlined}\\
&=\frac{pq^2}{2(M-p)}+\left[\frac{-pq^2}{2(M-p)}+O\left(\frac{pq}{M}\right)+\frac{O(p^2q^2)}{3(M-p)^2}\right]+O\left(\frac{p^2q^2}{M^2}\right)\\
&=O\left(\frac{pq}{M}\right),
\end{aligned}
\end{align}
which gives \eqref{eqn:tv-pq} by Pinsker's inequality \eqref{eqn:pinsker}.

The bound \eqref{eqn:tv-sparse} follows since for $\C^n$-valued random variables $X,Y$ and any measurable $h:\C^n\to\C^\ell$,
\begin{align*}
d_\TV(h(X),h(Y))&=\sup_{A}\left|\P(h(X)\in A)-\P(h(Y)\in A)\right|\\
&=\sup_A\left|\P(X\in h^{-1}(A))-\P(Y\in h^{-1}(A))\right|
\le d_\TV(X,Y).
\end{align*}
\end{proof}

\begin{acknowledgments}
The authors are extremely grateful to Professor Carlo Beenakker for providing the key reference  -- Ref.~\cite{FriedmanMello1985} -- for the COE submatrix density used in the proof and for clarifying the nature of the density distribution derived there. The authors would also like to acknowledge generous support by Michael Egley and sponsorship by the U.S. Air Force SEQCURE2 program at University of Maryland's Applied Research Laboratory for Intelligence and Security. VG was also supported by the US Army Research Office under Grant Number W911NF-23- 1-024 and Simons Foundation.
\end{acknowledgments}

\bibliographystyle{quantum}
\bibliography{hiding.bib}

\end{document}